\documentclass[a4paper,11pt]{article} 
  
\usepackage[english]{babel}
\usepackage[a4paper]{geometry}
\usepackage{amsmath}
\usepackage{amsthm}
\usepackage{amssymb}
\usepackage{bbm}
\usepackage{color}
\usepackage{enumerate}
\usepackage{mathrsfs}

\usepackage{hyperref}         

\def\bR{\mathbb{R}}

\def\bN{\mathbb{N}}

\def\bZ{\mathbb{Z}}
\def\cC{\mathcal{C}}

\def\cM{\mathcal{M}}

\def\cB{\mathcal{B}}

\def\cV{\mathcal{V}}
\def\cO{\mathcal{O}}

\def\cF{\mathcal{F}}
\def\cG{\mathcal{G}}
\def\cL{\mathcal{L}}
\def\cJ{\mathcal{J}}

\def\cN{\mathcal{N}}
\def\cE{\mathcal{E}}
\def\cK{\mathcal{K}}

\def\cH{\mathcal{H}}

\def\eps{\varepsilon}
\def\ph{\varphi}

\def\indic{\hbox{\raise-2pt \hbox{\indbf 1}}}

\let\==\equiv

\let\0=\noindent

\def\*{{\hfill\break\null\hfill\break}}

\def\tende#1{\,\vtop{\ialign{##\crcr\rightarrowfill\crcr
             \noalign{\kern-1pt\nointerlineskip}
             \hskip3.pt${\scriptstyle #1}$\hskip3.pt\crcr}}\,}
\def\otto{\,{\kern-1.truept\leftarrow\kern-5.truept\to\kern-1.truept}\,}

\def\tr{{\rm tr}}

\newtheorem{theorem}{Theorem}[section]  

\newtheorem{lemma}[theorem]{Lemma}

\numberwithin{equation}{section}

\def\be{\begin{equation}}
\def\ee{\end{equation}}

          \let\ph=\varphi   
        
        \let\L=\Lambda

\def \blue#1 {\textcolor{blue}{#1}}
\def \red#1 {\textcolor{red}{#1}}

\definecolor{lightblue}{rgb}{0, 0.33, 0.71}

\title{Bose gases in the Gross-Pitaevskii limit: \\ a survey of some rigorous results\footnote{Dedicated to Elliott Lieb on the occasion of his ninetieth birthday, with admiration.}} 

\author{Benjamin Schlein\footnote{Institute of Mathematics, University of Zurich, Winterthurerstrasse 190, 8057 Zurich.}} 

\begin{document} 

\maketitle

\begin{abstract}
We review some mathematical work on the Bose gas in the Gross-Pitaevskii regime. We start with the classical results by Lieb, Seiringer and Yngvason on the ground state energy and by Lieb and Seiringer on the existence of Bose-Einstein condensation. 
Afterwards, we discuss some more recent progress, based on a rigorous version of Bogoliubov theory. 
\end{abstract} 

\section{Trapped Bose gases: energy estimate and condensation} 


{\bf Bose-Einstein condensation for ideal gases.} The phenomenon of Bose-Einstein condensation has been first predicted by Einstein \cite{Ei} in 1924, based on previous work by Bose \cite{Bo} on the foundations of quantum statistical mechanics. Einstein considered the ideal Bose gas in a box $\Lambda_L = [0;L]^3$, with periodic boundary conditions. In the grand-canonical ensemble, the gas is described on the bosonic Fock space $\cF = \bigoplus_{n \geq 0} L^2 (\L_L)^{\otimes_s n}$, by the non-interacting Hamilton operator 
\begin{equation}\label{eq:HNideal} H = \sum_{p \in \L_L^*} p^2 a_p^* a_p \end{equation} 
where the sum runs over all momenta in $\L_L^* = (2\pi/L) \bZ^3$ and where $a_p^*, a_p$ denote the usual creation and annihilation operators on $\cF$, satisfying canonical commutation relations. The partition function of the gas at inverse temperature $\beta > 0$ and chemical potential $\mu$ is given by 
\[ Z (\beta, \mu) = \tr_\cF \; e^{-\beta (H - \mu \cN)} \]
where $\cN$ denotes the number of particles operator. Writing \[ H-\mu \cN = \sum_{p \in \L^*_L} (p^2 - \mu) a_p^* a_p \] 
and taking into account the fact that, for every momentum $p \in \L^*_L$, the spectrum of the harmonic oscillator $a_p^* a_p$ consists of all natural numbers, we find 
\[ Z (\beta ,\mu) = \prod_{p \in \L^*_L} \sum_{n =0}^\infty e^{- \beta n (p^2 - \mu)} = \prod_{p \in \L^*_L} \frac{1}{1-e^{-\beta (p^2 - \mu)}} \]
under the assumption that $\mu < 0$ (to make sure that the sum converges, for all $p \in \L_L^*$). The expected number of particles is then given by 
\[ \langle \cN \rangle = \frac{1}{Z (\beta, \mu)} \tr_\cF \, \cN e^{-\beta (H-\mu \cN)} = \frac{1}{\beta} \frac{d}{d\mu} \log Z (\beta, \mu) = \sum_{p \in \L_L^*} \frac{1}{e^{\beta (p^2 - \mu)} - 1} \, . \]
Hence, the chemical potential $\mu < 0$ can be fixed through the density 
\[ \rho = \frac{\langle \cN \rangle}{L^3} = \frac{1}{L^3} \sum_{p \in \L_L^*} \frac{1}{e^{\beta (p^2 - \mu)} -1} \, . \]
Let us now write $\rho = \rho_0 + \rho_+$, where
\begin{equation}\label{eq:rho0} \rho_0 = \frac{1}{L^3} \frac{1}{e^{-\beta \mu} - 1} \end{equation}
is the density of particles with zero momentum. In the limit $L \to \infty$, we can approximate 
\[ \rho_+ = \frac{1}{L^3} \sum_{p \in \L_L^* \backslash \{ 0 \}} \frac{1}{e^{\beta (p^2 - \mu)} -1} \simeq \frac{1}{(2\pi)^3} \int \frac{1}{e^{\beta (p^2 - \mu)} -1} dp = \frac{1}{(4 \pi \beta)^{3/2}} \sum_{n=1}^\infty \frac{e^{\beta \mu n}}{n^{3/2}} \, . \]
The r.h.s. of the last equation is increasing in $\mu \in (-\infty ; 0)$ and remains finite, as $\mu \to 0$. This implies that the density $\rho_+$ of particles with momentum $p \not = 0$ is always bounded by the critical density 
\[ \rho_c (\beta) =   \frac{1}{(4 \pi \beta)^{3/2}} \sum_{n=1}^\infty \frac{1}{n^{3/2}}  \, .\]
If $\rho > \rho_+ (\beta)$, we have to accommodate a macroscopic fraction of the particles in the state with zero momentum. In other words, we have to choose $\mu$ depending on $L$ with $\mu(L) \to 0$ so that, inserting in (\ref{eq:rho0}), we find $\rho_0 = \rho - \rho_c (\beta)$, in the limit $L \to \infty$ (we can choose $\mu(L) = - 1/ [\beta L^3 (\rho - \rho_c (\beta))]$). 

The accumulation of a macroscopic fraction of the particles in a single quantum state is known as Bose-Einstein condensation. For the ideal gas at inverse temperature $\beta > 0$, Bose-Einstein condensation takes place in the zero-momentum state whenever $\rho > \rho_c (\beta)$. Alternatively, we can invert the relation $\beta \to \rho_c (\beta)$ to define a critical temperature 
\[ \beta_c (\rho) = \frac{1}{4\pi \rho^{2/3}} \Big[ \sum_{n=1}^\infty \frac{1}{n^{3/2}} \Big]^{2/3}. \] 
For fixed density $\rho$, the ideal gas exhibits Bose-Einstein condensation at low temperature, when $\beta > \beta_c (\rho)$. At zero temperature (for $\beta = \infty$), we have complete condensation, ie. all particles are in the same zero-momentum state. 

\medskip

{\bf Condensation for interacting gases in the thermodynamic limit.} So far, we considered ideal gases. Establishing the existence of Bose-Einstein condensation for interacting Bose gases is much more challenging. In fact, the only available proof of condensation for interacting systems in the thermodynamic limit has been given by Dyson, Lieb and Simon in \cite{DLS}, for three-dimensional hard-core bosons on a lattice at half-filling, mapping the Bose gas to a spin system and using the tool of reflexion positivity that was previously introduced in \cite{FSS} (the approach of \cite{DLS} has been later extended by Kennedy, Lieb and Shastry in \cite{KLS} to any dimension greater than one and, more recently, by Aizenman, Lieb, Seiringer, Solovej and Yngvason in \cite{ALSSY} to certain models describing optical lattices). There has been, however, partial progress towards a better mathematical understanding of condensation for more general Bose gases in the thermodynamic limit through a renormalization group approach; see \cite{BFKT} for a review of recent results. 

\medskip

{\bf The Gross-Pitaevskii regime.} Since 1995, Bose-Einstein condensates have been realized in labs; see \cite{CW,K}. In typical experiments, a dilute Bose gas is initially trapped by strong magnetic fields. Then, the system is cooled down until it thermalizes at very low temperatures (the experimental realization of Bose-Einstein condensation was achieved through the development of cooling techniques, like laser cooling and evaporative cooling, which allowed physicists to reach temperatures in the range of nano-Kelvin; see \cite{Phi} and references therein). Afterwards, the traps are switched off so that the gas starts to evolve and the momentum distribution of the particles can be measured. At sufficiently low temperature, it turns out that most particles occupy the same one-particle orbital, showing condensation. 

From the point of view of mathematics, Bose-Einstein condensates produced in these experiments can be described as systems of $N$ bosons, trapped by an external 
potential with a characteristic length scale $L$ and interacting through a repulsive potential with scattering length $\frak{a} \ll L$. A particularly important scaling limit is the so called Gross-Pitaevskii regime, where $N \frak{a} = L$. In this regime, we have $\rho \frak{a}^3 = N \frak{a}^3 /L^3 = N^{-2}$; the Gross-Pitaevskii limit describes therefore very dilute Bose gases.

In the following, we are going to choose units so that $L = 1$. We consider $N$ bosons moving in the three-dimensional space $\bR^3$. The Hamilton operator is given by 
\begin{equation}\label{eq:HN} H_N = \sum_{j=1}^N \big[ -\Delta_{x_j} + V_\text{ext} (x_j) \big] + \sum_{i<j}^N N^2 V (N (x_i - x_j) \end{equation} 
and acts on the Hilbert space $L^2_s (\bR^{3N})$, defined as the subspace of $L^2 (\bR^{3N})$ consisting of functions that are symmetric with respect to permutations of the $N$ particles. In (\ref{eq:HN}), $V_\text{ext} : \bR^3 \to \bR$ is a confining potential (ie. $V_\text{ext} (x) \to \infty$, as $|x| \to \infty$) and $V : \bR^3 \to \bR$ is an interaction potential, which we assume to be radial, non-negative and to decay sufficiently fast at infinity (at least as $|x|^{-3-\eps}$, for an $\eps > 0$). 

\medskip

{\bf The scattering length.} We denote by $\frak{a}$ the scattering length of the potential $V$. Let us recall that the scattering length is a physical quantity measuring the effective range of $V$. It is defined through the solution of the zero-energy scattering equation
\begin{equation}\label{eq:0-en} \Big[ -\Delta + \frac{1}{2} V \Big] f = 0 \end{equation} 
with the boundary condition $f(x) \to \infty$, as $|x| \to \infty$. Asymptotically for $|x| \gg 1$, we find  
\[ f(x) \simeq 1 - \frac{\frak{a}}{|x|} \]
where the constant $\frak{a} > 0$ is the scattering length of $V$. For integrable potentials, it can be equivalently determined by the identity 
\begin{equation}\label{eq:scaV} 8\pi \frak{a} = \int V (x) f(x) dx \end{equation} 
which in particular shows that, for repulsive interactions, $8\pi \frak{a} < \widehat{V} (0)$. By scaling, it is then simple to check that the scattering length of the interaction potential $N^2 V(N.)$ appearing in (\ref{eq:HN}) is given by $\frak{a}/N$, which characterises the Gross-Pitaevskii regime.


\medskip

{\bf Leading order estimate for ground state energy.} The ground state energy $E_N (V_\text{ext},V)$ of the Hamilton operator (\ref{eq:HN}) has been first estimated by Lieb, Seiringer and Yngvason in \cite{LSY}. They proved that 
\begin{equation}\label{eq:LSY} \lim_{N \to \infty} \frac{E_N (V_\text{ext},V)}{N} = \inf_{\substack{\ph \in L^2 (\bR^3) :\\ \| \ph \| =1}} \cE_\text{GP} (\ph) \end{equation} 
with the Gross-Pitaevskii energy functional 
\begin{equation}\label{eq:GP} 
\cE_\text{GP} (\ph) =  \int \Big[ |\nabla \ph|^2 + V_\text{ext} |\ph|^2 + 4 \pi\frak{a} |\ph|^4 \Big] dx  \, . 
\end{equation} 
In particular, (\ref{eq:LSY}) implies that, to leading order, the ground state energy only depends on the interaction potential through its scattering length $\frak{a}$ (more recently, a different proof of (\ref{eq:LSY}) has been given in \cite{NRS}). 

The proof of (\ref{eq:LSY}) obtained in \cite{LSY} is based on matching upper and lower bounds for $E_N (V_\text{ext},V)$. The upper bound is achieved through an appropriate modification of the trial state used by Dyson in \cite{Dy} to estimate the energy of an homogeneous, dilute gas of hard spheres, in the thermodynamic limit. As for the lower bound, it relies on the techniques introduced by Lieb and Yngvason in \cite{LY}, extended to the non-homogeneous setting. 

\medskip

{\bf Dyson's Lemma.} A key ingredient in the proof of the lower bound is a lemma, originally due to Dyson, that can be used to replace the singular and short range interaction appearing in (\ref{eq:HN}) by a more regular potential, having smaller size and longer range. Since we believe it is interesting to compare Dyson's lemma with more recent tools that will be discussed in the next section, we state it precisely and we reproduce its short proof, in the form proposed by Lieb and Yngvason in \cite{LY}. 
\begin{lemma} \label{lm:dyson} 
Let $v$ be a non-negative function on $[0;\infty)$, with finite range $R_0 > 0$. Let $U$ be a non-negative function on $[0;\infty)$, with $U(r) = 0$ for all $r < R_0$ and with $\int U(r) r^2 dr  \leq 1$. Let $\cB \subset \bR^3$ be star shaped with respect to $0$. Then, we have 
\begin{equation}\label{eq:dyson} \int_\cB \Big[ \mu |\nabla \phi (x)|^2  + \frac{1}{2} v (|x|) |\phi (x)|^2 \Big] dx \geq \mu \frak{a} \int_\cB U (|x|) |\phi (x)|^2 dx \end{equation} 
where $\frak{a}$ is the scattering length (as defined in (\ref{eq:0-en})) of the potential $\bR^3 \ni x \to v (|x|)/\mu$. 
\end{lemma}
\begin{proof} 
It is enough to prove the claim for radial $\phi$. The general case follows because the total kinetic energy is always larger than its radial component. We consider first the case $U(r) = \delta (r-R)/R^2$ for an $R > R_0$ (this choice satisfies $\int U(r) r^2 dr \leq 1$). Writing $\phi (x) = u (|x|)/|x|$, we need to prove that 
\begin{equation}\label{eq:dys-cl} \int_0^{R_1} \Big[ \mu \big| u'(r) - u(r)/ r \big|^2 + \frac{1}{2} v(r) |u(r)|^2 \Big] dr \geq \left\{ \begin{array}{ll} 0 &\text{if } R_1 < R \\ \mu \frak{a} \frac{|u(R)|^2}{R^2} &\text{if } R_1 > R \end{array} \right.  \, .  \end{equation} 
By homogeneity in $u$, we can impose the normalization $u(R) = R-\frak{a}$. For $R_1 < R$, the inequality holds trivially. If $R_1 > R$, on the other hand, we can minimize the l.h.s. of (\ref{eq:dys-cl}) over all functions $u$, with $u(0)= 0$ and $u(R) = R- \frak{a}$. We find that the minimum is attained if $u = u_0$, with $u_0$ solving the zero-energy scattering equation 
\[ - \mu u''_0 (r) + \frac{1}{2} v(r) u_0 (r) = 0  \, .  \]
The normalization $u_0 (R) = R - \frak{a}$ implies that $u_0 (r) = r-\frak{a}$ for all $r > R_0$. Integrating by parts, we conclude that 
\[\begin{split}  
\int_0^{R_1}  \Big[ &\mu \big| u'(r) - u(r)/r \big|^2 + \frac{1}{2} v(r) |u(r)|^2 \Big] dr \\ &\geq 
\int_0^{R}  \Big[ \mu \big| u'_0 (r) - u_0 (r)/r \big|^2 + \frac{1}{2} v(r) |u_0 (r)|^2 \Big] dr \\ &= \mu u_0 (R) u'_0 (R)  - \mu u_0 (R)^2/R = \mu \frak{a} \big( 1 - \frak{a} /R \big) \geq \mu \frak{a}  \big( 1 - \frak{a} /R \big)^2 = \mu \frak{a} |u(R)|^2 / R^2 \, .  \end{split} \]
The lemma follows by noticing that a general $U$ can be written as linear combination of interactions having the form $\delta (r-R)/R^2$, ie. 
\[ U(r) = \int \frac{\delta (r- R)}{R^2} U(R) R^2 dR \]
where $\int U(R) R^2 dR \leq 1$ by assumption. 
\end{proof} 
Already Dyson \cite{Dy} applied a version of Lemma \ref{lm:dyson}  to derive a lower bound for the ground state energy of a dilute hard-sphere Bose gas. In his work, however, he sacrificed the full kinetic energy to replace the potential. As a result, his lower bound was off, approximately by a factor 14. Lieb and Yngvason, on the other hand, realized in \cite{LY} that it is more convenient to choose $\mu = 1-\eps$ in (\ref{eq:dyson}). Therefore, through Lemma \ref{lm:dyson}, they ended up with a nice potential and still with a bit of kinetic energy available. After dividing the volume in small boxes, this procedure allowed them to consider the regularized potential as a perturbation of the remaining kinetic energy and to apply rigorous perturbation theory (in the form of Temple's inequality) to achieve the correct lower bound. 



\medskip

{\bf Condensation in the Gross-Pitaevskii limit.} With standard tools in analysis, Lieb, Seiringer and Yngvason proved in \cite{LSY} the existence of a unique, positive and continuously differentiable minimizer $\ph_\text{GP}$ for the Gross-Pitaevskii energy functional (\ref{eq:GP}). The convergence (\ref{eq:LSY}) of the ground state energy per particle towards the minimum of (\ref{eq:GP}) suggests that the ground state of the Hamiltonian (\ref{eq:HN}) (and, more generally, many-body states with sufficiently small energy) should exhibit Bose-Einstein condensation in $\ph_\text{GP}$.

In contrast with the Hamilton operator for the ideal gas, eigenfunctions of (\ref{eq:HN}) are not products of one-particle orbitals. For this reason, some care is required when defining  the meaning of condensation. It turns out that Bose-Einstein condensation can be formulated in terms of reduced density matrices. We define the reduced one-particle density matrix $\gamma_N$ associated with a many-body wave function $\psi_N \in L^2_s (\bR^{3N})$ as the non-negative trace-class operator on the one-particle Hilbert space $L^2 (\bR^3)$ with the integral kernel 
\begin{equation}\label{eq:gamma} \gamma_N (x;y) = N \int dx_2 \dots d x_N \, \psi_N (x , x_2, \dots , x_N) \overline{\psi}_N (y,x_2, \dots , x_N)  \, . \end{equation} 
Here, we chose the normalization $\tr \, \gamma_N = N$. With this definition, we can express the expectation of the number of particles in a state $\ph \in L^2 (\bR^3)$ as 
\begin{equation}\label{eq:defBEC} \langle \psi_N , a^* (\ph) a (\ph) \psi_N \rangle = \langle \ph , \gamma_N \ph \rangle \, .  \end{equation} 
Therefore, we say that $\psi_N \in L_s^2 (\bR^{3N})$ exhibits Bose-Einstein condensation in the state $\ph \in L^2 (\bR^3)$ if $\langle \ph, \gamma_N \ph \rangle$ remains of order $N$, in the limit $N \to \infty$ (in this sense, Bose-Einstein condensation is actually a property of a sequence of wave functions $\{ \psi_N \}_{N \in \bN}$, rather than of a single $\psi_N$). 

In \cite{LS}, Lieb and Seiringer proved that the one-particle density matrix $\gamma_N$ associated with the ground state of the Hamilton operator (\ref{eq:HN}) is such that 
\begin{equation}\label{eq:BEC1} \frac{1}{N} \langle \ph_\text{GP} , \gamma_N \ph_\text{GP} \rangle \to 1 \end{equation} 
in the limit $N \to \infty$. According to (\ref{eq:defBEC}), this means that the ground state of (\ref{eq:HN}) exhibits complete Bose-Einstein condensation in the minimizer $\ph_\text{GP}$ of the Gross-Pitaevskii functional (\ref{eq:GP}); all particles, up to a fraction vanishing in the limit $N \to \infty$, are described by the same one-particle orbital $\ph_\text{GP}$. In \cite{LS2,NRS}, the convergence (\ref{eq:BEC1}) was extended, with different techniques, to any sequence of approximate ground states, ie. any sequence of wave functions $\psi_N$ satisfying 
\begin{equation}\label{eq:apprGS} \lim_{N \to \infty} \frac{1}{N} \langle \psi_N, H_N \psi_N \rangle \to \inf_{\substack{\ph \in L^2 (\bR^3) : \\ \| \ph \|_2 = 1}} \cE_\text{GP} (\ph)   \, . \end{equation}
Observe that (\ref{eq:BEC1}) does not imply that the product state $\psi_N^\text{prod}  = \ph_\text{GP}^{\otimes N}$ is a good approximation for the ground state of (\ref{eq:HN}). In fact, a simple computation shows that
\[ \frac{1}{N} \langle \psi_N^\text{prod}, H_N \psi_N^\text{prod} \rangle \simeq  \int \Big[ |\nabla \ph_\text{GP}|^2 + V_\text{ext} |\ph_\text{GP}|^2 + \frac{1}{2} \widehat{V} (0) |\ph_\text{GP}|^4 \Big]  dx > \cE_\text{GP} (\ph_\text{GP})  \]
where we used the remark following (\ref{eq:scaV}). Thus, the energy of the product state exceeds the ground state energy, even at leading order. Correlations among particles are crucial to approach the ground state energy (\ref{eq:LSY}); only at the level of reduced density matrices, correlations are integrated out and convergence (\ref{eq:BEC1}) can be established. 


\medskip

{\bf Proof of condensation.} The proof of (\ref{eq:BEC1}) obtained in \cite{LS} is based on two main ingredients. First of all, it relies on accurate energy estimates, which imply that
the part of the ground state kinetic energy contributing, in the limit of large $N$, to the quartic term in the Gross-Pitaevskii functional (\ref{eq:GP}) (ie. the difference between the many-body kinetic energy and the kinetic energy of the Gross-Pitaevskii minimizer $\phi_\text{GP}$) is localized in small regions of the configuration space, where two particles are close (much closer than they typically are). Denoting by $\psi_N \in L^2_s (\L^N)$ the normalized ground state wave function, we define, for fixed ${\bf X} = (x_2, \dots , x_N)$, 
\[ f_{\bf X} (x) = \frac{1}{\ph_\text{GP} (x)} \psi_N (x, {\bf X})  \, . \]
Introducing the set 
\[ \Omega_{\bf X} = \big\{ x \in \bR^3 : \min_{k=2,\dots , N}  |x- x_k| \geq N^{-1/3 -\delta} \big\} \, , \]
Lieb and Seiringer show in \cite{LS} that, for $\delta > 0$ small enough,  
\begin{equation}\label{eq:poin} \lim_{N \to \infty} \int d{\bf X} \int_{\Omega_{\bf X}} dx \, |\ph_\text{GP} (x)|^2 \, |\nabla_x f_{\bf X} (x)|^2 = 0  \, .  \end{equation}
Thus, for every fixed ${\bf X} = (x_2, \dots , x_N) \in \bR^{3(N-1)}$, most of the mass of $\ph_\text{GP} (x) \nabla_x f_{\bf X} (x)$ is contained in the set $\Omega_{\bf X}^c$, given by the union of small balls around the points $x_j$, $j=2,\dots , N$, whose total volume is of the order $N^{-3\delta}$ and converges to $0$, as $N \to \infty$. 
 
The second important ingredient is a generalized Poincar{\'e} inequality, which can be used, together with (\ref{eq:poin}), to show that $f_{\bf X}$ is almost constant and thus to conclude Bose-Einstein condensation, as stated in (\ref{eq:BEC1}). A non-standard Poincar{\'e} inequality is needed here because, despite the fact that $\Omega_{\bf X}^c$ has a small volume, the set $\Omega_{\bf X}$ may be quite strange, it does not even need to be connected. In general, smallness of the gradient $\nabla_x f_{\bf X}$ on non-connected sets with small complements does not imply that $f_{\bf X}$ is approximately constant (counterexamples can be easily constructed). It does, however, if the total $L^2$-norm of 
$\nabla_x f_{\bf X}$ is bounded; more details on this part of the proof can be found in \cite{LSann}.  

\medskip

In the next section, we will discuss, in some details, recent improvements of the estimate (\ref{eq:LSY})  on the ground state energy and of the bound (\ref{eq:BEC1}) on the number of particles in the condensate, in the Gross-Pitaevskii limit. Before doing so, let us briefly mention some other mathematical results concerning the Bose gas in the Gross-Pitaevskii regime that have been established in the last two decades (of course, the list is not complete). 

\medskip

{\bf The Gross-Pitaevskii functional for rotating gases.} Working in the moving coordinate system, a Bose gas rotating with angular velocity $\Omega \in \bR^3$ in the Gross-Pitaevskii regime can be described by the many-body Hamiltonian 
\begin{equation}\label{eq:Hrot} H^\text{rot}_N = \sum_{j=1}^N \big[ (i\nabla_{x_j} + A (x_j))^2 + W (x_j) \big]  + \sum_{i<j}^N N^2 V (N (x_i - x_j)) \end{equation} 
with $A(x) = (\Omega \wedge x)/2$ and $W(x) = V_\text{ext} (x) - (\Omega \wedge x)^2/4$. Mathematically, the rotation of the gas modifies the Hamilton operator like a constant magnetic field $\Omega$, with vector potential $A$. In order for (\ref{eq:Hrot}) to be bounded below, the external confining potential must grow fast enough, as $|x| \to \infty$, to compensate the centrifugal potential. Under the assumption that $W(x) \to \infty$ (more precisely, $W(x) \geq c \log |x| - C$, for constants $c,C>0$), Lieb and Seiringer proved in \cite{LS2} that, in the same spirit as in (\ref{eq:LSY}), the ground state energy $E_N^\text{rot} (\Omega, V_\text{ext})$ of (\ref{eq:Hrot}), defined minimizing the expectation of (\ref{eq:Hrot}) over all normalized bosonic wave functions, is such that 
\begin{equation}\label{eq:LS-rot} \lim_{N \to \infty} \frac{E^\text{rot}_N (\Omega, V_\text{ext},V)}{N} = \inf_{\substack{\ph \in L^2 (\bR^3) : \\ \| \ph \|_2 =1}} \int \Big[ |(i\nabla + A) \ph|^2 + W |\ph|^2 + 4\pi \frak{a} |\ph|^4 \Big]  dx  \, .  \end{equation} 
In contrast with the case $\Omega = 0$, the functional on the r.h.s. of (\ref{eq:LS-rot}) does not always have a unique minimizer. For this reason, Bose-Einstein condensation has to be stated in a slightly more complicated form, compared with (\ref{eq:BEC1}). Denote by 
$\Gamma$ the set of reduced one-particle density matrices obtained from sequences of approximate ground states of (\ref{eq:Hrot}), in the limit $N \to \infty$. Complete Bose-Einstein condensation, as established in \cite{LS2}, means here that the extremal 
points in $\Gamma$ coincide with the rank-one orthogonal projections into minimizers of the Gross-Pitaevskii functional (\ref{eq:LS-rot}) and that every element of $\Gamma$ can be expressed as a convex combination of these projections.

The behavior of the Gross-Pitaevskii functional (\ref{eq:LS-rot}) depends on the size of the angular velocity $\Omega$. Assume the confining potential $V_\text{ext}$ to be cylindrically symmetric. Then, if $|\Omega|$ is sufficiently small, (\ref{eq:LS-rot}) has a unique cylindrically symmetric minimizer. In a fixed reference frame, its velocity is zero. As already remarked in \cite{LSY-sup}, this can be interpreted as a manifestation of superfluidity (for sufficiently small $\Omega$, the gas does not react to the rotation). For larger values of $|\Omega|$, on the other hand, the minimizers of (\ref{eq:LS-rot}) form quantized vortices, breaking the axial symmetry and therefore leading to a continuous family of minimizers; this was shown in \cite{S-rot} (remark that there is no symmetry breaking in the ground state, if we considered (\ref{eq:Hrot}) over the full Hilbert space $L^2 (\bR^{3N})$, neglecting the bosonic statistics). In fact, the formation of quantized vortices in minimizers of (\ref{eq:LS-rot}) and of related functionals describing superfluids and superconductors has generated 
a substantial mathematical literature; for a review of this interesting subject, see \cite{Afta,BBH,CPRY,SaSe,Si}.

The case of a rotating Bose gas in an harmonic trap is particularly interesting. Here, there is a critical angular velocity $\Omega_c$, above which the gas is no longer confined. Sufficiently close to (but below) $\Omega_c$, the system is effectively projected to the lowest Landau level and can be described there by a Hamilton operator, modelling the bosonic analogue of the fractional quantum Hall effect; this was proven in \cite{LewS}. Also for this simplified Hamiltonian, convergence to the minimum of an appropriately defined Gross-Pitaevskii energy functional can be established in appropriate scaling regimes, see \cite{LSY-yrast}.

\medskip

{\bf The Gross-Pitaevskii limit in 2 dimensions.} In two dimensions, for Bose gases trapped in a volume of order one by a confining potential, the Gross-Pitaevskii limit is characterised by the condition that 
\begin{equation}\label{eq:gdef}  g = \frac{N}{| \log (N \frak{a}^2)|} \end{equation} 
remains fixed (or tends to a finite, non-vanishing, limit, as $N \to \infty$). This implies that the Gross-Pitaevskii regime can only be approached if the scattering length of the interaction potential $\frak{a}$ is exponentially small in the number of particles. For this reason, from the point of view of physics, the Gross-Pitaevskii limit in two dimensions is probably not as relevant as it is in three dimensions. From the point of view of mathematics, however, it is a well-defined regime and it is certainly the source of interesting and challenging questions. 

In \cite{LSY2}, Lieb, Seiringer and Yngvason considered a two-dimensional Bose gas confined by a trapping potential $V_\text{ext}$ and interacting through a repulsive two-body potential $V$ with scattering length $\frak{a}$. Assuming (\ref{eq:gdef}) to be fixed (positive and finite), they proved the analogue of (\ref{eq:LSY}), showing that the many-body ground state energy is such that, as $N \to \infty$,
\[ \frac{E_N (V_\text{ext}, V)}{N} \to \min_{\substack{\ph \in L^2 (\bR^3): \\ \| \ph \| =1}} \int \big[ |\nabla \ph|^2 + V_\text{ext} |\ph|^2 + 4 \pi g |\ph|^4 \big] dx  \, .  \]
Also in this case, the techniques of \cite{LS,LS2,NRS} can be used to show that the ground state of the Hamiltonian exhibits complete Bose-Einstein condensation. 

\medskip

{\bf Dynamics of Bose gases.} Instead of considering equilibrium properties of the Hamilton operator (\ref{eq:HN}) describing Bose gases in the Gross-Pitaevskii regime, it is also intersting to approach the problem from a dynamical point of view. Physically, studying dynamics means that we consider the time-evolution resulting from a change of the external fields. Motivated by the experiments described above, where the evolution of initially confined gases is observed after traps are switched off, it seems important to consider the solution of the time-dependent many-body Schr\"odinger equation 
\begin{equation}\label{eq:schr} i\partial_t \psi_{N,t} = \Big[ \sum_{j=1}^N -\Delta_{x_j} + \sum_{i<j}^N N^2 V (N (x_i - x_j)) \Big] \psi_{N,t} \end{equation} 
with initial data $\psi_{N,0} \in L^2_s (\bR^{3N})$ describing the ground state of the trapped Hamiltonian (\ref{eq:HN}). As explained in the previous paragraphs, $\psi_{N,0}$ exhibits Bose-Einstein condensation in the minimizer $\ph_\text{GP}$ of the Gross-Pitaevskii energy functional (\ref{eq:GP}). A natural question is therefore the following: Is Bose-Einstein condensation preserved by the time-evolution? In other words, does $\psi_{N,t}$ continue to exhibit condensation also for $t \not = 0$ and, if yes, in which one-particle state? 

This question was first addressed (for Bose gases in the Gross-Pitaevskii regime) in \cite{ESY1,ESY2,ESY3}, where it was proven that, for every fixed $t \in \bR$, the solution $\psi_{N,t} = e^{-i H_N t} \psi_{N,0}$ of (\ref{eq:schr}) does indeed exhibit complete Bose-Einstein condensation and that the evolution of the condensate wave function $\ph_t$ is determined by the time-dependent Gross-Pitaevskii equation
\begin{equation}\label{eq:GPtd} i\partial_t \ph_t = -\Delta \ph_t + V_\text{ext} \ph_t + 8\pi \frak{a} |\ph_t|^2 \ph_t \end{equation} 
with the initial data $\ph_{t=0} = \ph_\text{GP}$. This result shows that Gross-Pitaevskii theory can also be used to predict non-equilibrium properties of Bose gases at low energy. Different proofs of the convergence towards the solution of the nonlinear Gross-Pitaevskii equation (\ref{eq:GPtd}) have been later obtained in \cite{P,BDS}. More recently, convergence with optimal rates has been established in \cite{BS}. 


\medskip

{\bf The Gross-Pitaevskii regime at positive temperature.} Recently, the results (\ref{eq:LSY}) and (\ref{eq:BEC1}) have been extended to positive temperature in \cite{DSY,DS}. In the homogenous setting, with $N$ particles moving in the unit torus $\L = [0;1]^3$ (so that the density of the gas is $\rho = N$) and interacting through a repulsive potential with scattering length $\frak{a}_N = \frak{a}/N$, for temperatures comparable with the critical temperature of the ideal Bose gas (ie. $T \simeq \rho^{2/3} = N^{2/3}$), the free energy has been shown to have the form (for some $\eps > 0$ small enough)  
\begin{equation}\label{eq:free} F (\beta, N) = F_0 (\beta, N) + 4\pi \frak{a}_N \big(2 \rho^2 - \rho^2_0 \big) + \cO (N^{1-\eps}) \end{equation} 
where $F_0 (\beta,N)$ and $\rho_0 = \rho_0 (\beta, N)$ are the free energy and the  condensate density for the ideal gas (so that the first term on the r.h.s. of (\ref{eq:free}) is of the order $N^{5/3}$ and the second of order $N$). Moreover, the one-particle reduced density matrix associated with the thermal state at inverse temperature $\beta$ (or to any approximate sequence of thermal states) has been proven to approach (in the trace norm topology) the reduced one-particle density matrix associated with the free Gibbs state. In particular, this shows the existence of a phase transition for Bose-Einstein condensation in the Gross-Pitaevskii regime. 

\medskip

{\bf Additional references.} There are several review articles and lecture notes covering in more details most of the subjects discussed in this section. The interested reader can consult, for example, \cite{BPS,LSSY,Rou,Sch1,Sch2,Sei1,Sei2,Yng}.



\section{Bogoliubov theory in the Gross-Pitaevskii regime}

In the last 5 years there has been some progress in the mathematical understanding of the spectral properties of Bose gases in the Gross-Pitaevskii limit, beyond the leading order estimates (\ref{eq:LSY}), (\ref{eq:BEC1}). In this section, we are going to review some of these  new results, which are based on a rigorous version of Bogoliubov theory, developed in \cite{BBCS1,BBCS2,BBCS3,BBCS4}. 

To simplify our discussion, we will focus here on the homogeneous, translation invariant setting. We will therefore consider a system of $N$ particles moving in the three-dimensional unit torus 
$\L$, described by the Hamilton operator 
\begin{equation}\label{eq:HN2}
H_N = \sum_{j=1}^N -\Delta_{x_j} + \sum_{i<j}^N N^2 V (N(x_i -  x_j)) 
\end{equation} 
acting on $L^2_s (\L^N)$. From (\ref{eq:LSY}), the ground state energy $E_N (V)$ of (\ref{eq:HN2}) is such that 
\begin{equation}\label{eq:LSYti} \lim_{N \to \infty} \frac{E_N (V)}{N} =  4 \pi \frak{a}  \, . \end{equation} 
Moreover, from (\ref{eq:BEC1}) we conclude that every sequence $\psi_N \in L^2_s (\L^N)$ of approximate ground states of (\ref{eq:HN2}), satisfying $\langle \psi_N, H_N \psi_N \rangle / N \to 4 \pi \frak{a}$ exhibits complete Bose Einstein condensation, in the sense 
\begin{equation}\label{eq:BECti} \frac{1}{N} \langle \ph_0 , \gamma_N \ph_0 \rangle \to 1 \end{equation} 
into the zero-momentum state $\ph_0 \in L^2 (\L)$ defined by $\ph_0 (x) = 1$ for all $x \in \L$. In this section, we are going to sketch some of the main ideas in the proof of the following theorem, taken from \cite{BBCS3,BBCS4}, improving on the estimates (\ref{eq:LSYti}), (\ref{eq:BECti}), for a class of integrable interaction potentials. \begin{theorem} \label{thm:bogo}
Let $V \in L^3 (\bR^3)$, non-negative, radial and with compact support. 
\begin{itemize}
\item[i)] {\bf Optimal rate for condensation.} Let $\psi_N \in L^2_s (\L^N)$ be a sequence of approximate ground states, such that 
\[ \langle \psi_N, H_N \psi_N \rangle \leq 4 \pi \frak{a} N + \zeta \]
for $\zeta \ll N$. Then, denoting by $\gamma_N$ the one-particle reduced density matrix of $\psi_N$, defined as in (\ref{eq:gamma}), we find
\begin{equation}\label{eq:optBEC} N - \langle \ph_0, \gamma_N \ph_0 \rangle \leq C (\zeta+1) \end{equation} 
for a constant $C > 0$. 
\item[ii)] {\bf Precise estimate for ground state energy.} We have 
\begin{equation}\label{eq:ENGP} \begin{split} E_{N} (V) = \; &4\pi \frak{a} (N-1) + e_\Lambda \frak{a}^2 \\ & - \frac{1}{2}\sum_{p\in 2\pi \bZ^3 \backslash \{ 0 \}} \left[ p^2+8\pi \frak{a}  - \sqrt{|p|^4 + 16 \pi \frak{a}  p^2} - \frac{(8\pi \frak{a})^2}{2p^2}\right] + \cO (N^{-1/4}) \,  
\end{split}
\end{equation}   
where 
\begin{equation}\label{eq:eLambda0}
e_\Lambda = 2 - \lim_{M \to \infty} \sum_{\substack{p \in \bZ^3 \backslash \{ 0 \} : \\ |p_1|, |p_2|, |p_3| \leq M}} \frac{\cos (|p|)}{p^2} \end{equation}
and, in particular, the limit exists. 
\item[iii)] {\bf Excitation spectrum.} The spectrum of $H_N - E_N (V)$, below a threshold $\zeta > 0$, consists of eigenvalues having the form 
 \begin{equation}
    \begin{split}\label{eq:excGP}
    \sum_{p\in 2\pi \bZ^3 \backslash \{ 0 \}} n_p \sqrt{|p|^4+ 16 \pi \frak{a}  p^2}+ \cO (N^{-1/4} \zeta^3) \, 
    \end{split}
    \end{equation}
with $n_p \in \bN$ for all momenta $p \in 2\pi \bZ^3 \backslash \{ 0 \}$.
\end{itemize}
\end{theorem} 
{\bf Remarks.} 
\begin{itemize}
\item[1)] The estimate (\ref{eq:optBEC}) shows that the number of excitations of the Bose-Einstein condensate in a many-body state $\psi_N$ can be controlled by the difference between the energy of $\psi_N$ and the ground state energy. In fact, for the ground state vector $\psi^\text{gs}_N$, we obtain the precise estimate on the condensate depletion (number of excitations of the condensate) 
\[ N - \langle \ph_0 , \gamma_N^\text{gs} \ph_0 \rangle = \sum_{p \in 2\pi \bZ^3 \backslash \{0 \} } \Big[ \frac{p^2 + 8\pi \frak{a} - \sqrt{|p|^4 + 16 \pi \frak{a} p^2}}{2 \sqrt{|p|^4 + 16 \pi \frak{a} p^2}} \Big] + \cO (N^{-1/8}) \, .   \]
Non-optimal estimates establishing Bose-Einstein condensation have been also recently established in regimes interpolating between the Gross-Pitaevskii and the thermodynamic limit, in \cite{F}. 
\item[2)] The correction $e_\Lambda \frak{a}^2$, appearing in (\ref{eq:ENGP}), is a finite volume effect, arising because $\frak{a}$ is defined through the scattering equation (\ref{eq:0-en}) on the whole space $\bR^3$, rather than on the unit torus $\Lambda$.
\item[3)] The expression (\ref{eq:ENGP}) is the analogue, in the Gross-Pitaevskii regime, of the Lee-Huang-Yang formula for the energy of a dilute Bose gas in the thermodynamic limit. The validity of the Lee-Huang-Yang formula has been recently established through matching upper \cite{YY,BCS} and lower bounds \cite{FS1,FS2}.
\item[4)] Eq. (\ref{eq:excGP}) states that excited eigenvalues are determined, in good approximation, by the sum of the energies of  quantized excitations, labelled by their momentum and characterized by the dispersion law $\eps (p) = \sqrt{|p|^4 + 16 \pi \frak{a} p^2}$. 
\item[5)] Theorem \ref{thm:bogo} has been recently extended to Bose gases in the Gross-Pitaevskii regime, trapped by an external potential in  \cite{NNRT,NT,BSS1,BSS2}. In particular, the approach of \cite{NNRT,NT} shows how the assumption $V \in L^3 (\bR^3)$ appearing in Theorem~\ref{thm:bogo} can be relaxed to $V \in L^1 (\bR^3)$. 
\end{itemize}

Below, we sketch some of the main steps in the proof of Theorem \ref{thm:bogo}. We follow here the approach of \cite{BBCS1,BBCS2} (an alternative approach has been recently proposed in \cite{H,HST}). The proof of the optimal bound (\ref{eq:optBEC}) for the number of excitations and the proof of the estimates (\ref{eq:ENGP}) and (\ref{eq:excGP}) for the low-energy spectrum of (\ref{eq:HN2}) follow a similar strategy. We will discuss in  more details the proof of (\ref{eq:optBEC}). At the end, we will briefly comment on how (\ref{eq:optBEC}) can be used to deduce (\ref{eq:ENGP}) and (\ref{eq:excGP}). 

\medskip

{\bf Factoring out the Bose-Einstein condensate.} 
The first step to show Theorem \ref{thm:bogo} consists in factoring out the Bose-Einstein condensate and focussing instead on its orthogonal excitations. Following \cite{LNSS}, we remark that every $\psi_N \in L^2_s (\L^N)$ can be written as 
\[ \psi_N = \alpha_0 \ph_0^{\otimes N} + \alpha_1 \otimes_s \ph_0^{\otimes (N-1)} + \dots + \alpha_N \]
with $\alpha_j \in L^2_\perp (\L)^{\otimes j}$, for $j =0, \dots , N$. Here $L^2_\perp (\L)$ denotes the orthogonal complement in $L^2 (\L)$ of the condensate wave function $\ph_0$. This remark allows us to introduce a unitary operator $U_N$, defined by $U_N \psi_N = \{ \alpha_0 , \alpha_1 , \dots  , \alpha_N \}$, mapping the $N$-particle Hilbert space $L^2_s (\L^N)$ onto the truncated Fock space \[ \cF_+^{\leq N} = \bigoplus_{j=0}^N L^2_\perp (\L)^{\otimes j},\] where we describe orthogonal excitations of the condensate. 

With $U_N$, we can define the excitation Hamiltonian $\cL_N = U_N H_N U_N^*$, acting on $\cF^{\leq N}_+$. To compute $\cL_N$, we switch to momentum space and introduce the formalism of second quantization. For momenta $p,q \in \L^*_+ = 2\pi \bZ^3 \backslash \{ 0 \}$, we find that 
\begin{equation}\label{eq:rules} \begin{split} 
U_N a_p^* a_q U_N^* &= a_p^* a_q, \\ U_N a_0^* a_0 U_N^* &= N- \cN_+ \\
U_N a_p^* a_0 U_N^* &= a_p^* \sqrt{N- \cN_+} =: \sqrt{N} b_p^* , \\ U_N a_0^* a_p U_N &= \sqrt{N- \cN_+} a_p =: \sqrt{N} b_p \end{split} \end{equation} 
where $\cN_+$ denotes the number of particles operator on $\cF_+^{\leq N}$ (measuring the number of excitations of the condensate) and where we defined modified creation and annihilation operators $b_p^*, b_p$ creating and annihilating excitations with momentum $p$ (on states with few excitations of the condensate, we expect $b_p^\sharp \simeq a_p^\sharp$; in contrast to the standard creation and annihilation operators, $b_p, b_p^*$ preserve the total number of particles and thus map $\cF_+^{\leq N}$ back to itself). Writing 
\[ H_N = \sum_{p\in 2\pi \bZ^3} p^2 a_p^* a_p + \frac{1}{2N} \sum_{p,q,r \in 2\pi \bZ^3} \widehat{V} (r/N) a_{p+r}^* a_q^* a_{q+r} a_p \]
and applying (\ref{eq:rules}), we obtain the following lemma.
\begin{lemma}
The excitation Hamiltonian $\cL_N$ has the form 
\begin{equation}\label{eq:cLN} \begin{split} \cL_N = \; &\frac{(N-1)}{2N} \widehat{V} (0) (N- \cN_+) + \frac{\widehat{V} (0)}{2N} \cN_+ (N - \cN_+) \\ &+  \sum_{p \in \Lambda^*_+} p^2 a_p^* a_p + \sum_{p \in \Lambda_+^*} \widehat{V} (p/N) a_p^* \frac{N-\cN_+-1}{N} a_p + \frac{1}{2} \sum_{p \in \Lambda^*_+} \widehat{V} (p/N) (b_p^* b_{-p}^* + \text{h.c.} ) \\ &+ \frac{1}{\sqrt{N}} \sum_{p,q \in \Lambda_+^* , p+q \not = 0} \widehat{V} (p/N) \left[ b_{p+q}^* a_{-p}^* a_q + a_q^* a_{-p} b_{p+q} \right]  \\ &+ \frac{1}{2N} \sum_{p,q \in \Lambda^*_+ , r \not = p , -q} \widehat{V} (r/N) a_{p+r}^* a_q^* a_p a_{q+r} \, .  \end{split} \end{equation} 
\end{lemma}
Conjugation with the unitary map $U_N$ can be interpreted as a rigorous version of Bogoliubov c-number substitution, where operators $a^*_0 , a_0$, creating and, respectively, annihilating a particle in the condensate are replaced by a factor 
$\sqrt{N}$. 

\medskip

{\bf Renormalized excitation Hamiltonian.} Following Bogoliubov's approach, we should now approximate the excitation Hamiltonian (\ref{eq:cLN}) by a quadratic operator, neglecting cubic and quartic contributions appearing in the last two lines. While this  step leads to the correct energy in mean-field type regimes \cite{S,LNSS}, it cannot be justified in the Gross-Pitaevskii limit. Cubic and quartic terms in (\ref{eq:cLN}) still contain important contributions to the energy, which cannot be neglected (in fact, quartic terms even contribute to (\ref{eq:LSYti}) at leading order). 

The point here is that conjugation with $U_N$ factors out products of $\ph_0$ but keeps correlations in the excitation vector $U_N \psi_N = \{ \alpha_0, \dots , \alpha_N \}$. To extract relevant contributions from cubic and quartic terms in (\ref{eq:cLN}), we need to factor out the correlation structure characterizing low-energy states. 

To model correlations we fix $\ell_0 > 0$ (small but independent of $N$) and we consider the ground state solution of the Neumann problem 
\begin{equation}\label{eq:neu1} \Big[ -\Delta + \frac{N^2}{2} V (N x) \Big] f_{N} (x) = \lambda_{N} f_{N} (x) \end{equation} 
on the ball $|x| \leq \ell_0$, normalized so that $f_{N} (x) = 1$, if $|x| = \ell_0$. We extend $f_{N} (x) = 1$, for $|x| \geq \ell_0$, and we define $\check{\eta} (x) = - N (1-f_N (x))$. 
To understand the choice of $\check{\eta}$, we remark that, heuristically,  
\begin{equation}\label{eq:heuri} \begin{split}  \prod_{i<j}^N f_N (x_i - x_j) &= \prod_{i<j}^N \big[1 + \frac{1}{N} \check{\eta} (x_i - x_j) \big] \\ &\simeq \prod_{i<j}^N e^{\frac{1}{N} \check{\eta} (x_i - x_j)}  = e^{\frac{1}{N} \sum_{i<j} \check{\eta} (x_i - x_j)} \\ &\simeq \Big[ U^*_N \, e^{\frac{1}{2} \sum_{p \in \L^*_+} \eta_p b^*_p b_{-p}^*}   \Omega \Big] (x_1, \dots , x_N) \end{split} \end{equation} 
where $\{ \eta_p \}_{p \in \L^*_+}$ are the Fourier coefficients of $\check{\eta}$. 
To justify the last step in (\ref{eq:heuri}), we expand the exponential map, we use  (\ref{eq:rules}) to compute the action of $U^*_N$, and we switch between momentum 
and  position space. For technical reasons it is more convenient, instead of working with $\exp (\frac{1}{2} \sum_{p \in \L^*_+} \eta_p b^*_p b_{-p}^*)$, to use unitary operators and to introduce a cutoff $\kappa_H > 0$, restricting the sum in the exponent to large momenta. Thus, we define the generalized Bogoliubov transformation 
\begin{equation}\label{eq:Teta} 
T_\eta = \exp \Big[ \frac{1}{2} \sum_{|p| > \kappa_H} \eta_p \big( b_p^* b_{-p}^* - b_p b_{-p} \big) \Big]  \, . 
\end{equation} 
Motivated by (\ref{eq:heuri}), we expect that conjugation with $T_\eta$ factors out two-body correlations described by the solution $f_N$ of (\ref{eq:neu1}). For this reason, we define the renormalized excitation Hamiltonian $\cG_N = T_\eta^* \cL_N T_\eta = T^*_\eta U_N H_N U^*_N T_\eta$. 

On states with few excitations of the condensate, we have $b^\sharp_p \simeq a^\sharp_p$ and we can expect that the action of $T_\eta$ is close to action of a standard Bogoliubov transformation, ie. 
\begin{equation}\label{eq:TaT} \begin{split} T_\eta^* a_p^* T_\eta \simeq \cosh (\eta_p) \, a_p^* + \sinh (\eta_p) \, a_{-p} \\ 
T_\eta^* a_p T_\eta \simeq \cosh (\eta_p) \, a_p + \sinh (\eta_p) \, a^*_{-p} \end{split} \end{equation} 
for all $p \in \L^*_+$ with $|p| > \kappa_H$. With these formulas, we can approximately compute the renormalized excitation Hamiltonian $\cG_N$. 
\begin{lemma} \label{lm:GN} 
Introducing the notation 
\begin{equation}\label{eq:KV} \cK = \sum_{p \in \L^*_+} p^2 a_p^* a_p, \qquad \cV_N = \frac{1}{2N} \sum_{p,q \in \L^*_+, r \not = -p,-q} \widehat{V} (r/N) a_{p+r}^* a_q^* a_{q+r} a_p \end{equation} 
for the kinetic and potential energy operators, we find that 
\begin{equation}\label{eq:cGN} 
\begin{split} \cG_N = \; & 4\pi \frak{a} (N- \cN_+) + \big[ \widehat{V} (0) - 4\pi \frak{a} \big] \cN_+ \frac{(N-\cN_+)}{N} \\ &+  \cK + \sum_{p \in \L^*_+} \widehat{V} (p/N) a_p^* a_p (1 - \cN_+ / N) + 4\pi \frak{a} \sum_{|p| \leq \kappa_H} \big[ b_p^* b_{-p}^* + b_p b_{-p} \big]  \\ &+ \frac{1}{\sqrt{N}} \sum_{p,q \in \L^*_+ : p + q \not = 0} \widehat{V} (p/N) \big[ b_{p+q}^* a_{-p}^* a_q + \text{h.c.} \big] + \cV_N + \cE_{\cG_N}  \end{split} \end{equation} 
where the error operator $\cE_{\cG_N}$ is such that  
\begin{equation}\label{eq:cE} \pm \cE_{\cG_N} \leq C \kappa_H^{-1} (\cK + \cV_N)  + C \kappa_H  \, .  \end{equation}  
\end{lemma} 
Comparing (\ref{eq:cGN}) with the corresponding expression (\ref{eq:cLN}) for the excitation Hamiltonian $\cL_N$, we observe two main differences. First of all, the vacuum expectation of $\cG_N$ is given by $4\pi \frak{a} N$. To leading order, we know from (\ref{eq:LSYti}) that this is the correct ground state energy. Secondly, the off-diagonal quadratic term has been renormalized, meaning that the singular potential $\widehat{V} (p/N)$, decaying only for $|p|$ of order  $N$, has been now replaced by $4\pi \frak{a} \chi (|p| \leq \kappa_H)$, which decays already on momenta of order one (the cutoff $\kappa_H$ will be chosen large but fixed, independent of $N$). 

\medskip

{\bf Sketch of proof of Lemma \ref{lm:GN}}. Let us briefly explain the mechanism leading from (\ref{eq:cLN}) to (\ref{eq:cGN}). With (\ref{eq:TaT}) and approximating $\cosh \eta_p \simeq 1$, $\sinh \eta_p \simeq \eta_p$ for all $|p| > \kappa_H$ (higher order terms decay faster in $p$, hence they produce smaller errors), the conjugation of the kinetic energy operator gives 
\begin{equation}\label{eq:T-K} \begin{split} T_\eta^*\Big[  \sum_{p \in \L^*_+} p^2 a_p^* a_p \Big] T_\eta &\simeq \sum_{p \in \L_+^*} p^2  (\cosh (\eta_p) a_p^* + \sinh (\eta_p) a_{-p}) ( \cosh (\eta_p) a_p + \sinh (\eta_p ) a_{-p}^*) \\
&\simeq \sum_{p \in \L^*_+} p^2  a_p^* a_p + \sum_{|p| > \kappa_H} p^2 \eta_p (a_p^* a_{-p}^* + a_p a_{-p} \big) + \sum_{|p| > \kappa_H} p^2 \eta_p^2
 \end{split} \end{equation} 
(in this heuristic discussion, we exchange freely modified creation and annihilation operators with standard creation and annihilation operators). Analogously, we can compute the action of $T_\eta$ on the off-diagonal quadratic term in the second line on the r.h.s of (\ref{eq:cLN}). We find 
\begin{equation} \label{eq:T-L2} T_\eta^* \Big[ \frac{1}{2} \sum_{p\in \L^*_+} \widehat{V} (p/N) \big( a_p^* a_{-p}^* + \text{h.c.} \big) \Big] T_\eta \simeq \frac{1}{2} \sum_{p \in \L^*_+} \widehat{V} (p/N) (a_p^* a_{-p}^* + \text{h.c.} \big) + \frac{1}{2} \sum_{|p| \leq \kappa_H} \widehat{V} (p/N) \eta_p \end{equation} 
up to terms that can be included in the error operator $\cE_{\cG_N}$ in (\ref{eq:cGN}). Finally, we consider the conjugation of the quartic term on the r.h.s. of (\ref{eq:cLN}). In this case, the computation is a bit longer. At the end we find, up to irrelevant error terms, 
\begin{equation}\label{eq:T-V} \begin{split} 
T_\eta^* \Big[ \frac{1}{2N} &\sum_{p,q,r \in \L^*_+} \widehat{V} (r/N) a_{p+r}^* a_q^* a_{q+r} a_p  \Big] T_\eta \\  \simeq \; &\frac{1}{2N} \sum_{p,q,r \in \L^*_+}  \widehat{V} (r/N) a_{p+r}^* a_q^* a_{q+r} a_p  \\ &+ \frac{1}{2N} \sum_{\substack{q,r \in \L^*_+ : \\ |q+r| > \kappa_H}} \widehat{V} (r/N) \eta_{q+r} \big( a^*_q a^*_{-q} + a_{q} a_{-q} \big) + \frac{1}{2N} \sum_{\substack{q,r \in \L^*_+ : \\ |q|, |q+r| > \kappa_H}} \widehat{V} (r/N) \eta_{q+r} \eta_q
\end{split} \end{equation} 
Combining the quadratic terms generated in (\ref{eq:T-K}), (\ref{eq:T-L2}) and (\ref{eq:T-V}) we obtain, separating contributions associated with $|p| > \kappa_H$ and with $|p| \leq  \kappa_H$,  
\begin{equation}\label{eq:cancel2} \begin{split} 
\sum_{|p| > \kappa_H} \Big[ p^2 \eta_p &+ \frac{1}{2} \widehat{V} (p/N) + \frac{1}{2N} \sum_{\substack{r \in \L^* : \\ |p+r| > \kappa_H}} \widehat{V} (r/N) \eta_{p+r} \Big] \big( a_p^* a_{-p}^* + a_{-p} a_{p} \big) \\ &+ \frac{1}{2} \sum_{|p| \leq \kappa_H} \Big[ \widehat{V} (p/N) + \frac{1}{N} \sum_{\substack{r \in \L^* : \\ |p+r| > \kappa_H}} \widehat{V} (r/N) \eta_{p+r} \Big] \big( a_p^* a^*_{-p} + a_{-p} a_p \big) \, .   \end{split}  \end{equation} 
In the first line, the expression in the parenthesis (after extending the sum over $r$ to complete the convolution of $\widehat{V}$ and $\eta$, at the expenses of a small error) reconstructs exactly the l.h.s. of (\ref{eq:neu1}); using the identity (\ref{eq:neu1}), we can prove that this term is small and can be neglected. On the second line, on the other hand, we obtain, after completing the sum and switching to position space, the integral (\ref{eq:scaV}) (using the fact that the difference between $f_N$, defined in (\ref{eq:neu1}), and the solution of the zero energy scattering equation (\ref{eq:0-en}) is small). This explains the emergence of the renormalized off-diagonal quadratic term in (\ref{eq:cGN}). Similarly, combining the constant terms on the r.h.s. of (\ref{eq:T-K}), (\ref{eq:T-L2}), (\ref{eq:T-V}) with the constant term in (\ref{eq:cLN}), we obtain, using again the eigenvalue equation (\ref{eq:neu1}),  
\[ \begin{split} \frac{N}{2} \widehat{V} (0) + \sum_{|p| < \kappa_H} \Big[ p^2 \eta_p^2 + \widehat{V}(p/N) \eta_p &+ \frac{1}{2N} \sum_{\substack{r \in \L^* : \\\ |p+r| > \kappa_H} }\widehat{V} (r/N) \eta_{p+r} \Big] \\ &\simeq \frac{N}{2} \int V(x) \big[1 + \frac{1}{N} \check{\eta} (x) \big] dx \simeq 4\pi \frak{a} N ,  \end{split} \]
which explains the constant term on the r.h.s. of (\ref{eq:cGN}). More details on the proof of Lemma \ref{lm:GN} can be found in \cite[Sect. 7]{BBCS4}.

\medskip

{\bf Cubic renormalization.} Although the renormalized excitation Hamiltonian (\ref{eq:cGN}) looks simpler to handle, compared with (\ref{eq:cLN}), it is not yet coercive enough to prove Bose-Einstein condensation (it would be, if we assumed the interaction potential $V$ to be small enough, because then we could control the r.h.s. of (\ref{eq:cGN}) with kinetic and  potential energy operators, which are positive; this approach has been considered in \cite{BBCS1}). The main problem appears to be the cubic term on the r.h.s. of (\ref{eq:cGN}), which has been left unchanged by the action of the Bogoliubov transformation $T_\eta$. To regularize the cubic term, we need a second renormalization of the excitation Hamiltonian, this time through the exponential of a cubic operator. We define
\begin{equation}\label{eq:A} A = \frac{1}{\sqrt{N}} \sum_{|r| > \kappa_H, |v|< \kappa_L} \eta_r \big[ b_{r+v}^* a_{-r}^* a_v  - \text{h.c.} \big]  \end{equation} 
with an additional cutoff $\kappa_L < \kappa_H$ and we introduce the unitary operator 
$S = e^A$. With $S$, we construct the (twice) renormalized excitation Hamiltonian $\cJ_N = S^* \cG_N S$. 
\begin{lemma}\label{lm:JN}
We have  
\begin{equation}\label{eq:cJN} \begin{split} \cJ_N = \; &4\pi \frak{a} (N- \cN_+) + 4\pi \frak{a} \frac{(N-\cN_+)}{N} \\ &+ \cK + 8\pi \frak{a} \sum_{|p| \leq \kappa_H} a_p^* a_p \frac{(N-\cN_+)}{N} + 4\pi \frak{a} \sum_{|p| \leq \kappa_H} \big( b_p^* b_{-p}^* + b_p b_{-p} \big) \\ &+ \frac{8\pi \frak{a}}{\sqrt{N}} \sum_{|p| \leq \kappa_H , q\in \L^*_+} \big[ b_{p+q}^* a_{-p}^* a_q + \text{h.c.} \big] + \cV_N + \cE_{\cJ_N} \end{split} \end{equation} 
where $\cE_{\cJ_N}$ satisfies $\pm \cE_{\cJ_N} \leq (1/4) (\cK + \cV_N) + C$, if the cutoffs $\kappa_L < \kappa_H$ are chosen appropriately.
\end{lemma}

\medskip

{\bf Sketch of proof of Lemma \ref{lm:JN}.} We briefly discuss the derivation of (\ref{eq:cJN}). In contrast with the generalized Bogoliubov transformation $T_\eta$ (whose action can be approximately computed using the formulas (\ref{eq:TaT})), there is no explicit expression for the action of $S$. For this reason, we compute the renormalized excitation Hamiltonian $\cJ_N$ through a commutator  expansion having the form 
\[ \cJ_N = e^{-A} \cG_N e^A \simeq \cG_N + [ \cG_N , A] + \frac{1}{2} [ [\cG_N , A] , A]  + \dots  \, .  \]
Fortunately, it turns out that only the first and the second commutator produce relevant contributions, everything else can be neglected. Being $A$ cubic in creation and annihilation operators, its commutator with the kinetic energy operator, ie. $[\cK, A]$, is going to be cubic in creation and annihilation operators (because of the canonical commutation relation $[a_p , a_q^*] = \delta_{pq}$ and of the similar relation satisfied by the modified creation and annihilation operators). On the other hand, the commutator with the potential energy operator, $[\cV_N ,A]$ is quintic in creation and annihilation operators. Some of the contributions to this last commutator, however, are not in normal order. When we rearrange them into normal order, we produce cubic terms. The important observation, then, is that with the correct choice of the cubic operator $A$, the cubic contributions emerging from the commutators of $A$ with kinetic and potential energy operators combine with the cubic term, let us denote it by $\cC_N$, in $\cG_N$ and produce the renormalized cubic term appearing in (\ref{eq:cJN}). The algebra here is quite similar to the one discussed in (\ref{eq:cancel2}) for the cancellation in the off-diagonal quadratic terms (this is the reason why the coefficients in the cubic operator (\ref{eq:A}) are the same as the coefficients  appearing in the Bogoliubov transformation $T_\eta$; physically, this reflects the idea that correlations are just produced by two-body scattering events). 

While $[\cK , A]$ and $[\cV_N , A]$ combine with the cubic term $\cC_N$ in $\cG_N$ to produce the renormalized cubic term in (\ref{eq:cJN}), the commutator $[\cC_N , A]$ and the double commutator $[[\cK, A], A]$ are quartic in creation and annihilation operators. Similarly, the double commutator $[[\cV_N , A],A]$ is of order six (in this heuristic argument, we do not distinguish between standard and modified creation and annihilation operators). When we rearrange these terms in normal order, we obtain new quadratic and constant terms, which renormalize quadratic and constant terms on the r.h.s. of (\ref{eq:cGN}) to produce quadratic and constant terms in (\ref{eq:cJN}). More details on the proof of Lemma \ref{lm:JN} can be found in \cite[Sect. 8]{BBCS4}.

\medskip

{\bf Proof of condensation.} Observing the expression on the r.h.s. of (\ref{eq:cJN}), we remark that $\cJ_N$ has essentially the same form as the original excitation Hamiltonian (\ref{eq:cLN}), but with the singular interaction $\widehat{V} (p/N)$ (decaying only for very large momenta of order $N$) replaced by the regularized potential $\widehat{\nu}_p  = 8\pi \frak{a} \, \chi (|p| \leq \kappa_H)$. The only difference is that, in (\ref{eq:cJN}), the quartic term $\cV_N$ still depends on the original interaction. To complete the renormalization, we could conjugate $\cJ_N$ with a last unitary operator, given this time by the exponential of a quartic expression in creation and annihilation operators; this strategy has been considered in \cite{ABS}. 

In \cite{BBCS4}, on the other hand, a simpler strategy has been developed to prove Bose-Einstein condensation directly from (\ref{eq:cJN}), without quartic renormalization. This strategy is based on localization in the number of particles, a technique first introduced by Lieb and Solovej in \cite{LSol} (and then adapted in \cite{LNSS} to a setting similar to the one we are considering here). The idea is as follows. To prove a lower bound for $\cJ_N$ we can discard in (\ref{eq:cJN}) the positive quartic term $\cV_N$ (after using some of it to control part of the error $\cE_{\cJ_N}$). Furthermore, we can insert, by hand, the quartic term 
\begin{equation}\label{eq:V-ren}  \frac{1}{2N} \sum_{p,q \in \L^*_+ , r \not = -p , -q} \widehat{\nu}_r  \, a_{p+r}^* a_q^* a_{q+r} a_p  \end{equation} 
associated with the regularized potential; this produces an error that can be bounded (since the sum in (\ref{eq:V-ren}) runs only over $|r| < \kappa_H$, with $\kappa_H$ independent of $N$) by $C \cN_+^2 /N$. Comparing with (\ref{eq:cLN},) we find 
\begin{equation}\label{eq:cJN2} \cJ_N \geq U_N H_N^\text{ren} U_N^* - \frac{C}{N} \cN_+^2  - \frac{1}{2} \cK - C \end{equation} with the renormalized many-body Hamiltonian (here $\nu$ is the function on $\L$, with Fourier coefficients $\widehat{\nu}_p$ defined above) 
\begin{equation}\label{eq:HNren} H_N^\text{ren} = \sum_{j=1}^N -\Delta_{x_j} + \frac{1}{N} \sum_{i<j}^N \nu (x_i - x_j)  \, .  \end{equation} 

Using the positivity of $\widehat{\nu}$, we have 
\[ \begin{split} 
 0 &\leq \int_{\L^2} dx dy \, \nu (x - y) \Big[ \sum_{j=1}^N \delta (x- x_j) - N \Big] \Big[ \sum_{i=1}^N \delta (y - x_j) - N \Big] \\ &= \sum_{i,j= 1}^N \nu (x_i - x_j) - N^2 \widehat{\nu}_0  = 2 \sum_{i<j}^N \nu (x_i - x_j) + N \nu (0) - N^2 \widehat{\nu}_0  \end{split} \]
which implies (with $\widehat{\nu}_0  = 4\pi \frak{a}$ and $\nu (0)| \leq C$) that 
\[ H_N^\text{ren}  \geq 4\pi \frak{a} N + \sum_{j=1}^N ( -\Delta_{x_j} ) - C \]
and thus, from (\ref{eq:cJN2}), that 
\begin{equation}\label{eq:cM-fin} \cJ_N \geq 4\pi \frak{a} N + \frac{1}{2} \cK - \frac{C}{N} \cN_+^2 - C \, .  \end{equation} 
In order to control the negative term, quadratic in $\cN_+$, we choose smooth $g, h : \bR \to [0;1]$ with $g (s) = 1$ for all $s \leq 1$, $g (s) = 0$ for $s > 2$, such that $g^2 + h^2 = 1$. For $\delta > 0$ small enough, we introduce the notation $g_\delta = g (\cN_+ / \delta N)$, $h_\delta = h (\cN_+ / \delta N)$. Up to errors that are small in the limit $N \to \infty$, we can approximate
 \begin{equation} \label{eq:loc} \cJ_N \simeq g_\delta  \cJ_N g_\delta +  h_\delta  \cJ_N h_\delta  \, .  \end{equation} 
On the range of $g_\delta$, we have $\cN_+ < 2\delta N$. From (\ref{eq:cM-fin}), we find, for $\delta > 0$ small enough, 
 \begin{equation}\label{eq:bd-g}  
 \begin{split} g_\delta \cJ_N g_\delta  &\geq  g_\delta \big[ 4\pi \frak{a} N + \frac{1}{2} \cK - C\delta \cN_+ - C \big] g_\delta \\ &\geq \big[ 4\pi \frak{a} N + \frac{1}{4} \cK - C \big] g_\delta^2 \geq \big[ 4\pi \frak{a} N + c \cN_+ - C \big] g_\delta^2  \end{split} \end{equation}
for a $c> 0$ sufficiently small. On the range of $h_\delta$, on the other hand, we have $\cN_+ > \delta N$; as a consequence, here we cannot have complete Bose-Einstein condensation. Applying (\ref{eq:BEC1}), in the form proven in \cite{LS2,NRS}, we conclude by contradiction that there exists $c > 0$ such that 
 \[ h_\delta \cJ_N h_\delta \geq (4\pi \frak{a} + c) N h_\delta^2 \geq \big[ 4\pi \frak{a} N + c \, \cN_+ \big] h_\delta^2 .\] From (\ref{eq:loc}) and (\ref{eq:bd-g}), this implies that, for $c > 0$ small enough, 
 \[ \cJ_N \geq 4\pi \frak{a} N + c\,  \cN_+ - C  \, .  \]
 
 Suppose now that $\psi_N \in L^2_s (\L^N)$ is a sequence of approximate ground states, satisfying  
\begin{equation}\label{eq:appro-gs} \langle \psi_N, H_N \psi_N \rangle \leq 4\pi \frak{a} N + \zeta \end{equation} 
and let $\xi_N = S T_\eta U_N \psi_N$ denote the corresponding (renormalized) excitation vectors. Then, using (\ref{eq:cM-fin}) and the inequality $\cN_+ \leq C \cK$ (due to the spectral gap of the kinetic energy), we find  
\[ \begin{split} \langle \xi_N , \cN_+ \xi_N \rangle &\leq C \langle \xi_N, \cK  \, \xi_N \rangle \\ &\leq C \langle \xi_N , (\cM_N - 4\pi \frak{a} N ) \xi_N \rangle + C \\ &= C \langle \psi_N, (H_N - 4\pi \frak{a} )\psi_N \rangle + C \leq C (\zeta + 1)  \, . \end{split} \]
Since conjugation with $T_\eta$ and $S$ can only increase the expected number of excitations by a constant factor, we conclude that 
\begin{equation}\label{eq:BEC3} C (\zeta + 1) \geq \langle U_N \psi_N, \cN_+ U_N \psi_N \rangle = N - \langle \psi_N, a_0^* a_0 \psi_N \rangle = N - \langle \ph_0, \gamma_N \ph_0 \rangle \end{equation} 
which implies Bose-Einstein condensation with optimal rate, as stated in (\ref{eq:optBEC}).

It is interesting to compare the renormalization procedure based on conjugation with unitary operators that we used to show (\ref{eq:BEC3}) with Dyson's Lemma (Lemma \ref{lm:dyson}). Both techniques can be used to replace a singular interaction (large size, small range) with a regularized potential (smaller size, larger range) having the same scattering length. A clear advantage of Dyson's lemma is its simplicity. Also the renormalization procedure has some advantages, which play an important role in the proof of (\ref{eq:BEC3}). First of all, with the renormalization procedure, there is almost no loss of kinetic energy (some kinetic energy is needed to bound the error $\cE_{\cJ_N}$; what is important, however, is that although we only use a little bit of kinetic energy, we nevertheless construct the full renormalized potential, producing the energy $4\pi \frak{a} N$ on the r.h.s. of (\ref{eq:cM-fin})). Another advantage of the renormalization procedure is that it leads to a genuine many-body interaction, as the one appearing in (\ref{eq:HNren}). Dyson's lemma, on the other hand, works on the one-particle level and generates only an interaction between nearest neighbours. 
 
 \medskip
 
 {\bf A-priori bounds.} Eq. (\ref{eq:BEC3}) gives optimal control on the number of excitations of the Bose-Einstein condensate in approximate ground states satisfying 
 (\ref{eq:appro-gs}). In turn, inserting in (\ref{eq:cJN}), this implies a bound for the energy of the excitations. If we strengthen (\ref{eq:appro-gs}), assuming that \begin{equation}\label{eq:psichi} \psi_N = \chi (H_N \leq 4 \pi \frak{a} N + \zeta) \psi_N \, ,
 \end{equation}  
we can also establish a-priori estimates for the expectation of higher powers of the number of particles operator and of their product with kinetic and potential energy operators. More precisely, for a sequence $\psi_N \in L^2_s (\L^N)$ satisfying (\ref{eq:psichi}), it is possible to show that, for every $k \in \bN$, there exists $C > 0$ such that the excitation vector $\xi_N = T_\eta U_N \psi_N$ satisfies  
\begin{equation}\label{eq:apri} \big\langle \xi_N , (\cK + \cV_N + 1) ( \cN_+ + 1)^k \xi_N \rangle \leq C (\zeta + 1)^{k+1} \, .  \end{equation} 
In fact, the same estimate holds true if we replace $\xi_N$ by $S \xi_N = S T_\eta U_N \psi_N$, since cubic renormalization can only increase number and energy of excitations by constant factors, independent of $N$ (on the contrary, conjugation with $T_\eta$ changes the energy by order $N$). The proof of (\ref{eq:apri}) can be found in \cite[Prop. 4.1]{BBCS3}; we skip here the details.

\medskip

{\bf Ground state energy and excitation spectrum.} With the help of the strong a-priori bounds (\ref{eq:apri}), we can go back to the renormalized excitation Hamiltonian (\ref{eq:cGN}), defined after (\ref{eq:Teta}). As we did in (\ref{eq:cGN}), we can compute again $\cG_N$, this time keeping track of all terms, whose expectation does not vanish in the limit $N \to \infty$ (in (\ref{eq:cGN}), on the other hand, order one contributions are included in the error $\cE$). We find  
\[ \cG_N = C_N + Q_N + \cC_N + \cV_N + \delta_{\cG_N} \]
where $C_N$ is a constant term, $Q_N$ is quadratic, $\cC_N$ is cubic in creation and annihilation operators, and where $\delta_{\cG_N}$ is an error term, satisfying 
\begin{equation}\label{eq:cE-G} \pm \delta_{\cG_N} \leq \frac{C}{\sqrt{N}} (\cH_N + 1) (\cN_+ + 1)^2 \, .  \end{equation} 
From the a-priori estimates we know that the contribution of the error $\delta_{\cG_N}$ is negligible, in the limit of large $N$, on low-energy states. 

As already discussed in the proof of condensation, the contribution of the cubic term $\cC_N$ is relevant, it cannot be neglected. To get rid of the cubic term, we proceed as in the proof of condensation and perform a cubic renormalization. We define $\cJ'_N = e^{-A'} \cG_N e^{A'}$, for a cubic operator $A'$, similar to (\ref{eq:A}) (but not identical; for technical reasons, we find it convenient to choose 
$A'$ slightly different from the operator $A$ that we used in the proof of condensation). 
Computing $\cJ'_N$ up to errors that vanish in the limit $N \to \infty$, we find 
\begin{equation}\label{eq:cub-bog} \cJ'_N = C'_N + Q'_N + \cV_N + \delta_{\cJ'_N} \end{equation} 
where $C'_N$ and $Q'_N$ are new (renormalized) constant and cubic terms and where, similarly to (\ref{eq:cE-G}), 
\begin{equation}\label{eq:cE-bd}  \pm \delta_{\cJ'_N} \leq \frac{C}{N^{1/4}} (\cH_N + 1 ) ( \cN_+ + 1)^2  \, .  \end{equation} 
A part from the positive quartic term $\cV_N$, the Hamiltonian $\cJ'_N$ is essentially quadratic in creation and annihilation operators. Thus, it can be (approximately) diagonalized by a (generalized) Bogoliubov transformation $T_\tau$ (defined similarly as in (\ref{eq:Teta}), but of course with appropriate coefficients $\tau_p$). We define the final diagonalized excitation Hamiltonian $\cM_N = T_\tau^* \cJ'_N T_\tau$. We find 
\begin{equation}\label{eq:cMN} \begin{split}  \cM_N = \; &4\pi \frak{a} (N-1) + e_\Lambda \frak{a}^2 \\ & - \frac{1}{2}\sum_{p\in\Lambda^*_+} \left[ p^2+8\pi \frak{a}  - \sqrt{|p|^4 + 16 \pi \frak{a}  p^2} - \frac{(8\pi \frak{a})^2}{2p^2}\right] + \sum_{p \in \L^*_+}  \sqrt{|p|^4 + 16 \frak{a} p^2 } \, a_p^* a_p \\ &+ \cV_N + \delta_{\cM_N} 
\end{split} \end{equation} 
where the error operator $\delta_{\cM_N}$ satisfies a bound similar to (\ref{eq:cE-bd}).
To prove (\ref{eq:ENGP}) and (\ref{eq:excGP}) we need to show matching lower and upper bounds. For the lower bounds, the quartic interaction $\cV_N$ is positive and can be neglected. For the upper bounds, it turns out that the expectation of $\cV_N$ is small, on the appropriate trial states (determined by the quadratic part of (\ref{eq:cMN})). Hence, in (\ref{eq:cMN}) we can completely forget about $\cV_N$. Therefore, the estimates (\ref{eq:ENGP}), (\ref{eq:excGP}) follow easily from (\ref{eq:cMN}) and from the observation that the spectrum of the harmonic oscillator $a_p^* a_p$ consists of all natural numbers. This concludes the sketch of the proof of Theorem \ref{thm:bogo}. 

\vspace{.2cm} 

{\bf Acknowledgements.} The author gratefully acknowledges support from the European Research Council through the ERC Advanced Grant CLaQS. Additionally, he acknowledges partial support from the NCCR SwissMAP and from the Swiss National Science Foundation through the Grant ``Dynamical and energetic properties of Bose-Einstein condensates''.

\end{document}